\documentclass[
 reprint,
 amsmath,amssymb,
 prl
]{revtex4-2}
\usepackage[utf8]{inputenc}
\usepackage{graphicx}
\usepackage{dcolumn}
\usepackage{bm}
\usepackage{stmaryrd}
\usepackage{dsfont}
\usepackage{tikz}
\usepackage{tkz-euclide}
\usepackage{amsthm}
\usepackage{mathtools}
\usepackage{hyperref}
\usepackage{braket}
\usepackage{pgfplots}
\usepackage{enumerate}
\usepackage{booktabs}
\usepgfplotslibrary{external}
\newtheorem{theorem}{Theorem}
\newtheorem{lemma}[theorem]{Lemma}
\theoremstyle{definition}

\theoremstyle{remark}

\begin{document}

\preprint{APS/123-QED}

\title{Thresholds for post-selected quantum error correction from statistical mechanics}

\author{Lucas~H.~English}
 \email{lucas.english@sydney.edu.au}
\affiliation{%
School of Physics, University of Sydney, Sydney, New South Wales 2006, Australia
}%
\author{Dominic~J.~Williamson}
\thanks{Current address: IBM Quantum, IBM Almaden Research Center, San Jose, CA 95120, USA}
\affiliation{%
School of Physics, University of Sydney, Sydney, New South Wales 2006, Australia
}%
\author{Stephen~D.~Bartlett}
\affiliation{%
School of Physics, University of Sydney, Sydney, New South Wales 2006, Australia
}%

\date{\today}

\begin{abstract}
We identify regimes where post-selection can be used scalably in quantum error correction (QEC) to improve performance.  We use statistical mechanical models to analytically quantify the performance and thresholds of post-selected QEC, with a focus on the surface code.  Based on the non-equilibrium magnetization of these models, we identify a simple heuristic technique for post-selection that determines whether to abort without requiring a decoder.  Along with performance gains, this heuristic allows us to derive analytic expressions for  post-selected conditional logical thresholds and abort thresholds of surface codes. We find that such post-selected QEC is characterised by four distinct thermodynamic phases, and detail the implications of this phase space for practical, scalable quantum computation.

\end{abstract}

\maketitle

\textit{Introduction.}---Quantum computers are expected to solve certain problems that are intractable with conventional methods \cite{Daley2022,doi:10.1126/science.aar3106}.  
However, quantum information is easily corrupted by interactions with the environment, and current qubit and gate error rates are too high to allow for large-scale quantum computations without some means to correct these errors \cite{RevModPhys.95.045005}. By using quantum error correcting codes and the techniques of fault-tolerant quantum computing, large-scale quantum computations with arbitrary accuracy are predicted to be possible \cite{Campbell2017, 2211.07629}. The overheads needed for quantum error correction, in terms of qubit redundancy and additional circuit complexity, can be very large, and there is a need for new approaches to reduce these overheads in practice \cite{2409.17595,Bravyi2024,RevModPhys.87.307}.

Post-selection has been proposed as a means to enhance the performance of quantum error correction by aborting, and potentially reinitializing, an experiment if the success of correction is sufficiently uncertain~\cite{cdi_proquest_journals_3054660994,cdi_proquest_journals_3052222419,ChenEdwardH2022CDfE,PhysRevA.110.012419,cdi_doaj_primary_oai_doaj_org_article_fe27d63449cc48129d59ac0f6bebc5ad,10.5555/2011763.2011764}. 
Knill first proposed such an enhancement in Ref.~\cite{KnillE2005Qcwr}, demonstrating improved threshold estimates for concatenated distance-2 codes and posing the question: \textit{are thresholds of post-selected computing  strictly higher than those of standard quantum computing?}
Recently, numerical and experimental results have demonstrated improved conditional logical error rates from post-selection \cite{cdi_proquest_journals_2899513284,HarperRobin2019FLGi,SundaresanNeereja2023Dmsq,PostlerLukas2022Dofu,BluvsteinDolev2024Lqpb,YeYangsen2023LMSP,GuptaRiddhiS2024Eams}. Despite these promising results, the theoretical limits of performance gains enabled by post-selected QEC are poorly understood.

Models from statistical mechanics provide a powerful set of tools for quantifying these limits.  For a QEC code, the value of the threshold with optimal decoding can be reformulated as the location of a phase transition in a statistical mechanical model~\cite{cdi_crossref_primary_10_1063_1_1499754}. This formulation was originally proposed for the surface code~\cite{cdi_crossref_primary_10_1063_1_1499754} and was later generalized to arbitrary stabilizer codes~\cite{cdi_crossref_primary_10_4171_AIHPD_105}. Recently, this approach has been extended to connect the stability of topologically-ordered phases under local perturbations with the thresholds of topological quantum error correcting codes~\cite{2406.15757}.

In this Letter, we examine how post-selection affects the mapping of Pauli stabilizer codes to disordered statistical mechanical models. Specifically, it limits the quenched average over all $n$-qubit Pauli errors to a subset defined by the post-selection rule. Although conceptually simple, its significance lies in the nontrivial definition of the accepted set. We demonstrate that optimal post-selection can be performed by setting a minimum cutoff value for the largest disordered partition function, which represents a logical coset probability under a maximum-likelihood decoder (MLD).
We propose a simple heuristic post-selection rule that does not require any decoding prior to aborting, and analytically compute bounds on its conditional logical thresholds. 
We apply a mean-field argument to analytically determine the abort threshold of the heuristic post-selection rule. 
We identify four distinct thermodynamic phases for post-selected quantum error correction, and identify regimes where post-selection can be advantageous in scalable fault-tolerant quantum computing. 
Finally, we apply the heuristic post-selection rule to experimental data taken from a recent QEC demonstration by Google~\cite{2408.13687}.

\textit{Statistical mechanical mapping.}---Pauli stabilizer codes can be mapped to classical spin Hamiltonians where each stabilizer check is represented by a spin $s_{k}=\pm 1$ and Pauli errors $E\in\mathcal{P}^{\otimes n}$ introduce quenched disorder~\cite{cdi_crossref_primary_10_1063_1_1499754,cdi_crossref_primary_10_4171_AIHPD_105}. In this mapping, the Hamiltonian is given by
\begin{equation}\label{eq:Hamiltonian}
    H_{E}(\vec{s})=-\sum_{i,\sigma\in\mathcal{P}_{i}}\overbrace{J_{i}(\sigma)}^\text{Strength}\overbrace{\llbracket\sigma,E\rrbracket}^\text{Disorder}\overbrace{\prod_{k:\llbracket\sigma,S_{k}\rrbracket=-1}s_{k}}^\text{Interactions},
    \end{equation}
where $J_{i}(\sigma)$ are noise-dependent interaction strengths, $\mathcal{P}_{i}$ is the local Pauli group on qubit~$i$ and the scalar commutator $\llbracket \sigma,E\rrbracket=\pm 1$ distinguishes ferromagnetic from antiferromagnetic couplings. Fixing a noise model $\mathbb{P}(E)$, a phase transition is identified when the quenched average free energy $[\langle F \rangle_{Z[E]}]_{E}$ becomes non-analytic. Along the Nishimori line \cite{cdi_crossref_primary_10_1143_PTP_66_1169}, where the conditions
\begin{equation}
    \beta J_{i}(\sigma)=\frac{1}{|\mathcal{P}|}\sum_{\tau\in\mathcal{P}_{i}}\text{log}\mathbb{P}_{i}(\tau)\llbracket\sigma,\tau^{-1}\rrbracket
\end{equation}
hold, the phase transition coincides with the code’s logical threshold under maximum-likelihood decoding \cite{cdi_crossref_primary_10_4171_AIHPD_105}. Moreover, a gauge symmetry of the Hamiltonian ensures that the disordered partition function $Z_{E}=\sum_{\vec{s}}e^{-\beta H_{E}(\vec{s})}$ encodes the error’s logical coset probability \cite{cdi_crossref_primary_10_4171_AIHPD_105}. For instance, the toric code under independent depolarizing noise maps to a disordered eight-vertex model \cite{cdi_doaj_primary_oai_doaj_org_article_be091cd1287f4f45881801af0e724334}, while under bit-flip noise it maps to a random-bond Ising model \cite{cdi_crossref_primary_10_1063_1_1499754}.

Post-selection partitions the total set of syndromes into an accept and an abort partition. However, as the classical spin Hamiltonian is parametrized by the Pauli error acting on the code, we instead treat post-selection as partitioning the total set of errors ${\mathcal{E}\in\mathcal{P}^{\otimes n}}$ into an accept and an abort partition, ${\mathcal{E}=\mathcal{E}_{\mathrm{accept}}\cup \mathcal{E}_{\mathrm{abort}}}$. 
By constraining errors which produce equivalent syndromes to lie in the same partition, we induce a canonical isomorphism between the partition of errors and syndromes.
Under the statistical mechanical mapping, the code's conditional logical threshold after post-selection can be determined by taking the quenched average over only the accepted set $\mathcal{E}_{\mathrm{accept}}$.

We define a post-selection rule $R$ by mapping a post-selection parameter $\gamma\in[0,1]$ to the abort partition of errors, $R(\gamma)=\mathcal{E}_{\mathrm{abort}}$.
Here, we set the convention that $|R(\gamma)|$ is nondecreasing with $\gamma$, $R(0)=\emptyset$, and $R(1)=\mathcal{E}$ is the maximal set of errors to abort upon given the post-selection rule. 
We define the map from the abort partition to the critical probability under a given decoder $\Pi: \{\mathcal{P}^{\otimes n}\} \rightarrow \mathbb{R}$, through $\Pi[\mathcal{E}_{\mathrm{abort}}]=p_{\mathrm{th}}^{\gamma}$. 
The composition of these maps may be taken to get a map from the post-selection parameter to the conditional logical threshold, $\Pi[R(\gamma)]=p_{\mathrm{th}}^{\gamma}$.

\textit{Optimal post-selection.}---QEC codes can be optimally decoded via MLD, which yields the probabilities for each logical error coset. Under the statistical mechanical mapping, these probabilities are given by the disordered partition functions $Z_{E}$ for representative errors $E$ from each logical class. For instance, in a surface code with a single logical qubit and a syndrome $S$, there are four inequivalent logical cosets. Defining the correction operators as $\{C_{i}\}_{i=1}^{4}$ with $C_{i}=C_{0}L_{i}$ (where $C_{0}$ satisfies $S$ and the $L_{i}$ span the logical space) captures this structure.

Under MLD of a QEC code, an abort parameter $\gamma$ can be chosen such that the decoder aborts if ${\max_{E\in \tilde{E}(S)}Z_{E}<\gamma}$, equivalently ${\max_{E\in \tilde{E}(S)}\mathds{P}(\bar{E})<\gamma}$, where $\tilde{E}(S)$ is the set of errors consistent with the observed syndrome $S$. Under MLD, the quenched average of the maximum disordered partition function, e.g., $[\max_{\sigma\in\mathcal{P}^{\otimes k}}Z_{\sigma E}]_{E}$ for a code with $k$ logical qubits, is the probability of decoding success, as the disordered partition functions correspond to the probabilities of each logical coset.

\begin{theorem}\label{thm:1}
    Post-selection that aborts if MLD returns a maximum coset probability less than some $\gamma\in[0,1]$ is optimal in the following sense. Such post-selection partitions the errors $\mathcal{E}$ into an abort set $\mathcal{E}_{\mathrm{abort}}$ of some cardinality $|\mathcal{E}_{\mathrm{abort}}|=r$. This partition strictly upper bounds $\mathbb{P}_{\mathrm{succ}}$ for any abort set with equivalent cardinality $r$.
\end{theorem}
\noindent A proof of Theorem 1 is given in the Supplemental Material.

\textit{Heuristic post-selection.}---As MLD for a general stabilizer code is \#P-Complete \cite{IyerPavithran2015HoDQ}, it remains computationally infeasible to perform optimal post-selection for quantum codes. 
Inspired by the statistical mechanical mapping and Refs.~\cite{cdi_doaj_primary_oai_doaj_org_article_fe27d63449cc48129d59ac0f6bebc5ad, cdi_proquest_journals_3052222419,cdi_proquest_journals_2899513284}, we propose a heuristic post-selection technique for surface codes. 
Given the statistical mechanical model of a quantum error correcting code, the non-equilibrium magnetization is $m\coloneq \frac{1}{|\mathcal{S}|}\sum_{k}s_{k}$, where the spin degrees of freedom correspond to the stabilizer measurements.
This magnetization is a good order parameter for the thermodynamic phases considered in this work, as shown in Ref.~\cite{cdi_crossref_primary_10_4171_AIHPD_105}, because it quantifies the density of violated stabilizers, which is anti-correlated with the logical success probability.
We now define a simple heuristic for post-selection of a surface code via the post selection rule
\begin{equation}
    R_{\text{heuristic}}(\gamma)=\{E\ |\ m(E)<-1+2\gamma\},
\end{equation}
where $m(E)$ denotes the non-equilibrium magnetization of the stabilizer measurements given some Pauli error~$E$. 
When $\gamma=0$, the set $\mathcal{E}_{\mathrm{abort}}=\emptyset$ and when $\gamma=1$, the set $\mathcal{E}_{\mathrm{abort}}$ is the set of errors that produce non-trivial syndromes. 
This heuristic post-selection rule has the advantage of simplicity over the logical gap rules considered in Refs.~\cite{cdi_proquest_journals_3052222419, cdi_doaj_primary_oai_doaj_org_article_fe27d63449cc48129d59ac0f6bebc5ad,cdi_proquest_journals_2899513284}. 
Specifically, the decision of whether to abort does not require a decoder, but rather only requires the density of stabilizers measuring $-1$ on the code.
We detail the conditions of error channels suitable for this rule in Appendix A.

We now apply the statistical mechanical mapping to determine the conditional logical thresholds under the heuristic post-selection rule for surface codes. When ${\gamma=0}$ we never abort, and the statistical mechanical Hamiltonian's behaviour is determined by taking the quenched average over all possible errors $E\in\mathcal{P}^{\otimes n}$. For a pure bit-flip or pure phase-flip channel, toric and surface codes are mapped to a random-bond Ising model (RBIM)~\cite{cdi_crossref_primary_10_1063_1_1499754}, which leads to a MLD threshold of $p_{\mathrm{th}}^{\gamma=0}\approx 10.94\%$ \cite{HoneckerA2001Ucot} and a MWPM threshold of $p_{\mathrm{th}}^{\gamma=0}\approx 10.31\%$ \cite{WangChenyang2003Ctia}.

When $\gamma=1$, we abort whenever the syndrome is non-trivial. 
In this case, the set $\mathcal{E}_{\mathrm{accept}}$ corresponds to the product of stabilizer generators and a chosen set of logical representatives $\{ \prod_{k}S_{k}L_{m}\}$.
The abort partition is the set of errors which produce non-trivial syndromes. 
By the symmetry of the Hamiltonian, we have $Z_{S_{l}}=Z_{S{m}}$ for all $S_{l},S_{m}\in \mathcal{S}$. When the quenched average free energy is computed, above threshold, the free energy cost of a non-trivial logical operator converges to a constant (see Eq. (36) of Ref.~\cite{cdi_crossref_primary_10_1063_1_1499754}). Hence, the non-trivial logical cosets do not contribute to order parameters for this conditional logical threshold, and we may consider only the identity coset.
For the identity coset, we have $Z_{IS_{k}}=Z_{I}$ for all stabilizers $S_{k}$, and we note $Z_{I}$ corresponds to the non-disordered (or ``clean" \cite{2406.15757}) Hamiltonian $H_{I}$. For each of the partition functions, the underlying Hamiltonian may be considered equivalent across all disorder parameters, and is given by $H_{I}$.

For the disordered random-bond Ising model, with $E=I$, the scalar commutators are uniformly $+1$, and the Hamiltonian becomes the non-disordered Ising model. The Ising model can be exactly solved to find $\beta J_{c}=\frac{\log(1+\sqrt{2})}{2}$ \cite{OnsagerLars1944CSIA}. Substituting this value into the Nishimori conditions, we arrive at $p_{\mathrm{th}}^{\gamma=1}=\frac{1}{2+\sqrt{2}}\approx0.2929$. While this gives the maximum post-selected conditional logical threshold under MLD, the decoding strategy is simply to do nothing regardless of the decoder that is used away from full post-selection, such as minimum-weight perfect matching (MWPM). We discuss thresholds for depolarizing noise in Appendix B.

For $0<\gamma<1$, the post-selected conditional logical threshold lies between the two bounds computed above.

\textit{Abort threshold.}---The existence of an abort threshold is an important necessary condition for post-selected QEC to be scalable. The existence of abort thresholds were first numerically observed in Ref.~\cite{cdi_proquest_journals_3052222419}. Abort thresholds are the phenomenon in which in the thermodynamic limit, physical error probabilities above a critical value result in aborting with certainty, while error probabilities below the critical value result in acceptance with certainty. The abort thresholds of the heuristic post-selection rule can be accurately obtained using a mean field approximation on the spins, which we derive in Appendix C.

\begin{figure}[t]
    \centering
    \includegraphics[width=\linewidth]{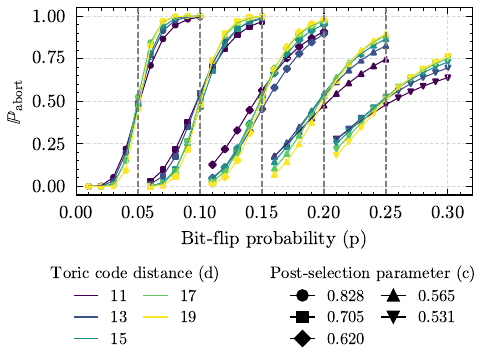}
    \caption{Abort thresholds of toric codes under bit-flip error channels. Dashed lines indicate analytic predictions using the mean field approximation. Markers indicate numerical results.}
    \label{fig:abort}
\end{figure}

We compute the expected magnetization under the mean field approximation given bit-flip probabilities of ${p=\{0.05, 0.1, 0.15, 0.2, 0.25\}}$, and sample $N=10^{7}$ errors from these noise channels acting on toric codes of varying sizes, computing abort probabilities given a post-selection parameter equal to the expected magnetization at these probabilities. The results are shown in Fig.~\ref{fig:abort}, demonstrating threshold behaviour near the bit-flip probabilities predicted.

Following the discussion in this section, similar abort thresholds can be computed analytically for any i.i.d. noise model, e.g., depolarizing noise, straightforwardly by computing the expected magnetization. 
For the optimal post-selection, the abort threshold is fixed at the code's non-post-selected logical threshold.

\textit{Numerical simulations.}---We now investigate the post-selected conditional logical thresholds under our heuristic post-selection rule. We expect that the post-selected conditional logical threshold will increase monotonically with the post-selection parameter $\gamma$. We analyze toric codes with parameters $[[2L^{2},2,L]]$. We use the Metropolis-Hastings algorithm for a Markov-Chain Monte Carlo (MCMC) method \cite{HastingsW.K.1970MCsm} to sample the conditional probability distribution $\mathbb{P}(E|m(S)>m_{0})$, and we decode the errors using PyMatching \cite{higgott2023sparse}. We use the fitting ansatz $\mathbb{P}_{\mathrm{fail}}=A+Bx+Cx^{2}$ with $x=(p-p_{\mathrm{th}})d^{\frac{1}{\nu{0}}}$ \cite{WatsonFernHE2014Lers} to extract an estimate of the conditional threshold from the results and indicate the estimate with a dashed line. The results are presented in Fig.~\ref{fig:thresholds}.

\begin{figure}[t]
    \centering
    \includegraphics[width=\linewidth]{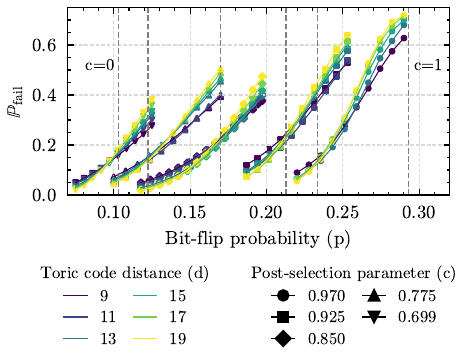}
    \caption{Conditional post-selected logical thresholds for post-selection parameter $0.6987\lessapprox \gamma<1$. Notation in the plot is similar to Fig.~\ref{fig:abort}.}
    \label{fig:thresholds}
\end{figure}

The post-selected conditional logical threshold increases with the post-selection parameter. In other words, as we constrain the non-equilibrium magnetization of the statistical mechanical model more tightly, we increase the phase transition to a higher disorder probability. Moreover, for $\gamma\lessapprox 0.6987$, the logical threshold of the code is unchanged. This critical value can be determined by the expected magnetization at $p=p_{\mathrm{th}}^{\gamma=0}\approx 0.103$. This point arises as a multi-critical point on the phase diagram, which we explain in Appendix D.

\textit{Thermodynamic phases.}---Through post-selection parameterized by $\gamma\in[0,1]$, we obtain conditional thresholds of $p_{\mathrm{log},\mathrm{th}}(\gamma)$ and $p_{\mathrm{abort},\mathrm{th}}(\gamma)$ for the logical error rates and probability of aborting, respectively. Under the heuristic post-selection rule with a bit-flip channel, we expect that the conditional logical thresholds between the two bounds (i.e., at $\gamma=0$ and $\gamma=1$) vary continuously. However, they are difficult to numerically analyze on small-scale code examples, as densities of non-trivial stabilizer measurements on finite codes are discrete and the cardinality of the stabilizers of small code instances of differing sizes generally do not have common factors. Nonetheless, we expect in the thermodynamic limit that clear thresholds emerge.

We plot in Fig.~\ref{fig:threshold_diagram} the logical and abort thresholds as a function of the post-selection parameter $\gamma$. We conjecture that the shaded region in Fig.~\ref{fig:threshold_diagram} corresponds to useful scalable post-selection.  Working in this region is below both abort and logical thresholds, meaning logical failure rates and abort rates can be made arbitrarily small by increasing the code size, but with the advantage that the logical threshold is increased over the non-post-selected case. If a system is above the logical threshold without post-selection, applying post-selection will not bring it below the conditional logical threshold without exceeding the abort threshold.  Instead, we emphasise that, although post-selection can increase the logical threshold, the benefits of this increased logical threshold are best harnessed by working with finite-sized codes well below this logical threshold. By assuming a finite-size scaling ansatz of the form $\mathbb{P}_{\mathrm{fail}}\sim f((p-p_{c})d^{\frac{1}{\nu}})$, post-selection can increase $p_{c}$, thereby shifting the scaling function $f$ further from criticality \cite{WatsonFernHE2014Lers}.  There is also strong numerical evidence that post-selection further improves the sub-threshold scaling~\cite{cdi_proquest_journals_3052222419}.  Our conjecture emphasizes the importance of sub-threshold behaviour in finite-sized codes for practical applications, rather than focusing solely on raising the threshold for noisy devices.

\begin{figure}[t]
    \centering
\includegraphics[width=0.75\linewidth]{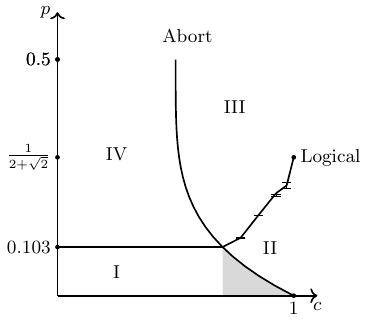}
    \caption{Phase diagram of toric code under the heuristic post-selection rule with bit-flip error channel. The x-axis represents the post-selection parameter $\gamma$. When $\gamma=0$, no post-selection is applied and when $\gamma=1$, we fully post-select, aborting whenever the syndrome is non-trivial. The y-axis represents the i.i.d.~bit-flip probability acting on physical qubits. The abort phase boundary is derived using a mean-field approximation of the non-equilibrium magnetization of the spins. The logical phase boundary is trivial for $\gamma\lessapprox 0.6987$ and the multi-critical point can be determined by the expected magnetization at $p=p_{\mathrm{th}}^{\gamma=0}$. Monte-Carlo simulations are used to estimate the non-trivial boundary. Four thermodynamic phases are identified: I. below abort and logical threshold, II. above abort threshold and below logical threshold, III. above abort and logical threshold, IV. below abort threshold and above logical threshold. The shaded area indicates a regime in which the conditional logical threshold is increased, while remaining below abort threshold. Error bars indicate $\pm$ one standard deviation of the parameter estimate of threshold probability.}
    \label{fig:threshold_diagram}
\end{figure}

\textit{Post-selected QEC in practice.}---We now apply the heuristic post-selection method to open source data reported from experiments where a surface-code memory was initialized on a superconducting processor~\cite{2408.13687}. While our theory employs an idealized bit‑flip noise model for analytic tractability, the statistical rationale in Appendix A suggests this heuristic can be usefully extended to random local circuit‑level noise.  Its effectiveness may depend on many factors including the code, circuit, schedule, and the relevant set of logical faults (e.g.,~in generalized lattice surgery) that extend beyond our present analysis.  We use the Google experimental data as an exploratory test of the usefulness of our heuristic in this setting, and find that our rule can improve logical error rates ex post facto, supporting its practical utility.
An XZZX surface code \cite{BonillaAtaidesJPablo2021TXsc} was initialized with distances 3, 5, and 7.
We apply the heuristic post-selection to these results for a single round of error correction.
Fig.~\ref{fig:google_plot} shows the abort and logical failure probabilities.

For no post-selection or moderate levels of post-selection $\gamma=0,0.85,0.9$, logical failure and abort probabilities decrease as the code distance increases, consistent with sub-threshold scaling for each metric (i.e., Phase I behaviour). As the post-selection parameter increases to $\gamma=0.95$, logical failure probability still decays with distance, but the abort probability grows instead, consistent with Phase II scaling.
Linear regression was applied to the log-linear plots, followed by hypothesis testing to determine whether the gradient was positive (above threshold) or negative (below threshold). Notably, 6 out of 7 regressions yielded one-sided p-values below 1\%, providing statistical evidence for this sub-threshold versus above-threshold scaling behaviour.

\begin{figure}[t]
    \centering
    \includegraphics[width=\linewidth]{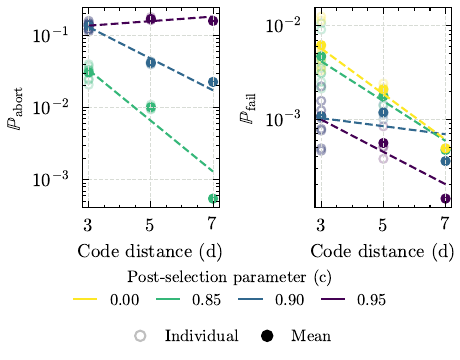}
    \caption{(Left) Abort probabilities and (right) logical failure probabilities for a single round of error correction on XZZX surface codes initialized on a superconducting quantum computer using the heuristic post-selection \cite{2408.13687}. We use the optimized decoding results reported in Ref.~\cite{bausch2023learningdecodesurfacecode} that were produced by a neural network decoder with reinforcement learning. We apply the heuristic post-selection rule on this data ex post facto to estimate abort probabilities and conditional logical failure probabilities. Transparent markers indicate individual data of surface codes initialized on different physical qubits. Solid markers indicate averages over all instances of distance $d$ codes. An exponential curve as a function of $d$ is fit to each set of data.}
    \label{fig:google_plot}
\end{figure}

\textit{Conclusion.}---Our results point to several promising directions for future work.
Our proposed heuristic post-selection technique offers practical advantages without the need for resource-intensive decoding. This can facilitate real-time decision-making and reducing computational overhead in post-selected quantum error correction.
The post-selected statistical mechanical models we have studied suggest new avenues of research in constrained classical spin models, whereby the ordered-to-disordered phase transition can be deformed with a post-selection parameter.
While our analytic treatment explores the code capacity setting, we have demonstrated the heuristic’s applicability to realistic circuit-level noise in an exploratory example using Google’s experimental data. A comprehensive study of performance under circuit-level noise is left for future work.

\textit{Acknowledgments.}---We thank Samuel Smith and Christopher Chubb for discussions.  This work is supported by the Australian Research Council via the Centre of Excellence in Engineered Quantum Systems (EQUS) project number CE170100009, and by the ARO through the QCISS program W911NF-21-1-0007 and IARPA ELQ program W911NF-23-2-0223.

\section{End Matter}
\textit{Appendix A: Effectiveness of heuristic post-selection rule.}---The effectiveness of the heuristic post-selection rule originates from the anti-correlation between the density of violated stabilizers and the logical success probability. This anti-correlation is not merely an empirical observation for simple noise models but stems from a general statistical argument that holds for a broad class of realistic, circuit-level noise channels including Pauli errors, gate faults, measurement errors, and other locally correlated, non-adversarial processes. The argument proceeds as follows: 

In any fault-tolerant implementation of a local stabilizer code, physical faults, whether Pauli errors, gate errors, or measurement errors, generate localized syndrome violations in the full spacetime history of the computation. Consequently, a high density of syndrome violations is a sufficient (but not necessary) condition to infer a high density of underlying physical faults. A single experimental run that exhibits a high density of syndrome violations is statistically difficult to distinguish from a typical run generated by a much noisier process. We can think of such a run as being drawn from a channel with a higher effective error rate, $p_{\mathrm{eff}}>p$. Our simple, decoder-free heuristic acts as an efficient, albeit imperfect, proxy for this underlying fault density. While optimal post-selection utilizing a maximum-likelihood decoder on the spacetime graph will certainly outperform this simple heuristic, a high spatial density is a strong indicator of a run that is likely to fail. It is a foundational feature of error correction that for a fixed-size code, the logical error rate is a monotonically increasing function of the physical error rate. By post-selecting against runs with high syndrome density (our proxy for high $p_{\mathrm{eff}}$), we are systematically discarding the runs that are likely to fail. This biases the accepted ensemble towards a lower effective error rate, thereby reducing the overall logical failure probability. While this principle is general, our analytic results in the main text focus on independent error channels (specifically bit-flip noise) for their tractability, which allows us to derive quantitative thresholds and phase diagrams. However, the rule may fail for highly correlated or adversarial noise channels, which would require tailored approaches beyond the scope of this work.

While our analytic thresholds and phase diagrams in the main text are derived under independent bit-flip noise for tractability, the same local syndrome density argument applies to general circuit-level noise. The success of our heuristic on Google’s surface code data, which contains full circuit-level noise and correlated error effects, supports its robustness beyond the idealized model, validating its potential as a practical tool. Although circuit‑level simulations of surface codes are routinely performed, a numerical study of our heuristic’s performance is challenging because it targets the tail of the syndrome distribution at low physical error rates and larger code distances. Sampling enough high syndrome density runs to gather meaningful statistics would require prohibitive computational resources and specialized importance‑sampling techniques, so we leave a full numerical exploration to future work.

\textit{Appendix B: Surface codes under depolarizing noise.}---The toric or surface codes with depolarizing noise without post-selection (i.e., $\gamma=0$) map to the disordered eight-vertex model, and the threshold can be approximated numerically as in Ref.~\cite{cdi_doaj_primary_oai_doaj_org_article_be091cd1287f4f45881801af0e724334}. 
Under MLD, this corresponds to a threshold of $p_{\mathrm{th}}^{\gamma=0}\approx18.9\%$~\cite{cdi_doaj_primary_oai_doaj_org_article_be091cd1287f4f45881801af0e724334}.
Under MWPM decoding, the threshold for this code is instead $p_{\mathrm{th}}^{\gamma=0}\approx15.5\%$ \cite{10.5555/2011362.2011368}. In the fully post-selected regime, the Hamiltonian becomes the non-disordered isotropic eight-vertex model, which can be solved exactly to give $\beta J_{c} = \frac{1}{4}\text{log}(3)$ \cite{alma991011130779705106}. 
If we substitute this value into the Nishimori conditions, we can rearrange to arrive at $p_{\mathrm{th}}^{\gamma=1}=0.5$, consistent with the numerical estimate in Ref.~\cite{cdi_proquest_journals_3052222419}.

\textit{Appendix C: Abort threshold derivation.}---To obtain the mean field approximation, we compute the expected magnetization in the thermodynamic limit of the spins assuming no inter-spin correlations. For example, consider a toric code under bit-flip noise with probability $p$. Each $\sigma_{z}$-type stabilizer has support on four qubits, each of which has a probability $p$ of suffering a bit-flip. 
Then, the probability that measuring the stabilizer $S_{k}$ returns $-1$ is ${\mathbb{P}(s_{k}=-1)={4 \choose 1}p(1-p)^{3} + {4 \choose 3}p^{3}(1-p)}$, while the probability that it returns $+1$ is
${\mathbb{P}(s_{k}=+1)={4 \choose 0}(1-p)^{4} + {4 \choose 2}p^{2}(1-p)^{2} + {4 \choose 4}p^{4}}$. Then, the expected value of any particular stabilizer's individual magnetization is $m(p)=\mathbb{E}[m_{s_{k}}] = -4(p(1-p)^{3} + p^{3}(1-p)) + (1-p)^{4} + 6p^{2}(1-p)^{2} + p^{4}$.

As each qubit lies in the support of two of each type of stabilizer (e.g., two plaquettes), the probabilities of the spin degrees of freedom are not independent.
However, for $p\ll 1$, we may approximate the stabilizers as being independent using a mean-field approximation. 
Under this assumption, in the thermodynamic limit, by the law of large numbers, the total system's magnetization tends towards the expected value given above. 
If we fix the post-selection parameter $m_{0}=-1+2\gamma$, then in the thermodynamic limit, for $m(p)<m_{0}$ we almost surely abort, and for $m(p)>m_{0}$, we almost surely accept, leading to threshold behaviour.

\textit{Appendix D: Multicritical point.}---Under the mean-field approximation, the expected non-equilibrium magnetization of the surface code under bit-flip noise possesses a $\mathbb{Z}_{2}$ symmetry, as it is invariant under the transformation $p\rightarrow 1-p$, and in fact on the domain $p\in[0,1]$, $\mathbb{E}[m_{s_{k}}]\geq 0$. Moreover, in the thermodynamic limit, for $\gamma<\gamma_{\mathrm{crit}}$, where $\gamma_{\mathrm{crit}}=\gamma\ |\ \mathbb{E}[m(\gamma)]=p_{\mathrm{th}}^{\gamma=0}$, we never abort at logical threshold. Because of this, we expect the conditional logical threshold to be unchanged. However, for $\gamma>\gamma_{\mathrm{crit}}$, at $p=p_{\mathrm{th}}^{\gamma=0}$, we are above the abort threshold, and thus the location of this phase transition can change non-trivially. Thus, in Fig.~\ref{fig:threshold_diagram}, we find the separation of phases I and IV to be parallel to the x-axis, while the separation of phases II and III to be non-trivial.

\newpage
\pagebreak
\clearpage
\widetext
\appendix

\begin{center}
\textbf{\large Supplemental Material for “Thresholds for post-selected quantum error correction from statistical mechanics”}\\[6pt]
Lucas H.~English,\; Dominic J.~Williamson,\; Stephen D.~Bartlett\\
\emph{School of Physics, University of Sydney, Sydney, New South Wales 2006, Australia}\\
(Dated: \today)
\end{center}
\vspace{1em}

\section{S1. Recap of Pauli Stabilizer Codes}

Quantum error correcting codes are a mapping from a \textit{logical} Hilbert space onto a larger \textit{physical} Hilbert space. The logical degrees of freedom are encoded in a subspace of the physical Hilbert space, $\mathcal{H}_{\mathrm{log}}\subseteq\mathcal{H}_{\mathrm{phy}}$. We say a code has parameters $[[n,k,d]]$ if $\dim(\mathcal{H}_{\mathrm{phy}})=2^{n}$, $\dim(\mathcal{H}_{\mathrm{log}})=2^{k}$ and $d=\min\{w(L)\ |\ L\neq I, L\ket{\psi}\in\mathcal{H}_{\mathrm{log}}$ for all $\ket{\psi}\in\mathcal{H}_{\mathrm{log}}\}$. That is, the distance $d$ is the smallest Hamming weight of a non-trivial logical operator.

Pauli stabilizer codes rely on the stabilizer formalism, in which we specify the logical subspace as the space stabilized by a group of operators we call the stabilizer group $S$. That is, $s\ket{\psi}=+1\ket{\psi}$ for all $s\in S$ and all $\ket{\psi}\in\mathcal{H}_{log}$. For Pauli stabilizer codes, elements of the stabilizer group are tensor products of Pauli operators, that is $s=\sigma_{i_{1}}\otimes\sigma_{i_{2}}\otimes\cdots\otimes\sigma_{i_{n}}\in\mathcal{P}^{\otimes n}$.

Under i.i.d. Pauli noise channels, the probability distribution of errors $E\in\mathcal{P}^{\otimes n}$ takes the form
\begin{equation}
    \mathds{P}(E)=\prod_{i=1}^{n}\mathds{P}(E_{i}),
\end{equation}
that is, the product of the probabilities of Paulis affecting each qubit independently.

\section{S2. Proof of Theorem 1}

\begin{lemma}
Post-selection which aborts when the posterior confidence under MLD is less than some constant $\gamma\in[0,1]$ monotonically increases decoding likelihood with respect to $\gamma$.
\end{lemma}

\begin{proof}
    We first show that post-selection as above for some $\gamma\in[0,1]$ lower bounds the probability of decoding success. Without post-selection we have
    \begin{equation}
        \mathbb{P}_{\mathrm{succ}}=\sum_{S}\mathbb{P}(S)\mathbb{P}(\mathrm{succ}|S),
    \end{equation}
    where $S$ denotes the measured syndrome, $\mathbb{P}(\mathrm{succ}|S)=\max_{\sigma\in\mathcal{P}^{\otimes k}}Z_{\sigma E(S)}$, with $E(S)$ being a representative error consistent with the syndrome $S$. By post-selecting, we partition the set $S$ into $S=S_{\mathrm{abort}}\cup S_{\mathrm{accept}}$, where the new probabilities over $S_{\mathrm{accept}}$ are normalized through
    
    \begin{equation}
        \mathbb{P}(S\in S_{\mathrm{accept}})=\frac{\mathbb{P}(S)}{\sum_{S\in S_{\mathrm{accept}}}\mathbb{P}(S)}.
    \end{equation}
    
    Assuming that $\mathbb{P}(\mathbb{P}(\mathrm{succ}|S)\geq\gamma)\neq 0$, the conditional probability of decoding success given a post-selection parameter $\gamma$ is
    \begin{align}
        \mathbb{P}_{\mathrm{succ}}^{\gamma}&=\sum_{S\in S_{\mathrm{accept}}}\mathbb{P}(S)\mathbb{P}(\mathrm{succ}|S),\\
        &\geq \sum_{S\in S_{\mathrm{accept}}}\mathbb{P}(S) \gamma,\\
        &= \gamma\sum_{S\in S_{\mathrm{accept}}}\mathbb{P}(S),\\
        &= \gamma.
    \end{align}

    That is, the probability of decoding success is lower bounded by $\gamma$. Note that we can trivially upper bound this probability also through aborting whenever the uncertainty is less than some constant $\gamma'$. Next, we show that by post-selection with parameter $\gamma$ increases the probability of decoding success without post-selection. The expectation of $\mathbb{P}(\mathrm{succ})$ over the entire set of syndromes $S$ can be expressed using the law of total expectation:

    \begin{align}
        \mathbb{E}[\mathbb{P}(\mathrm{succ})] &= \overbrace{\mathbb{E}[\mathbb{P}(\mathrm{succ})|\mathbb{P}(\mathrm{succ})\geq\gamma]}^{\geq\gamma}\mathbb{P}(\mathbb{P}(\mathrm{succ})\geq\gamma) + \overbrace{\mathbb{E}[\mathbb{P}(\mathrm{succ})|\mathbb{P}(\mathrm{succ})< \gamma]}^{< \gamma}\mathbb{P}(\mathbb{P}(\mathrm{succ})< \gamma),\\
        &\leq \mathbb{E}[\mathbb{P}(\mathrm{succ})|\mathbb{P}(\mathrm{succ})\geq\gamma].
    \end{align}

    So $\mathbb{P}_{\mathrm{succ}}\leq \mathbb{P}_{\mathrm{succ}}^{\gamma}$ with equality if and only if $\mathbb{P}(\mathbb{P}(\mathrm{succ}< \gamma))=0$, that is, all syndromes lead to a posterior confidence greater than or equal to $\gamma$. Finally, we show that by increasing the post-selection, we may increase the probability of decoding success. First we note that for two constants $\gamma,\gamma'\in[0,1]$ with $\gamma>\gamma'$,

    \begin{equation}
        \{S|\mathbb{P}(\mathrm{succ}|S)\geq\gamma\}\subseteq \{S|\mathbb{P}(\mathrm{succ}|S)\geq\gamma'\}.
    \end{equation}

    As with above, we expand the expectation value of $\mathbb{P}_{\mathrm{succ}}^{\gamma'}$ through:

    \begin{align}
        \mathbb{E}[\mathbb{P}(\mathrm{succ})|\mathbb{P}(\mathrm{succ})\geq\gamma']&=\mathbb{E}[\mathbb{P}(\mathrm{succ})|\mathbb{P}(\mathrm{succ})\geq\gamma]\mathbb{P}(\mathbb{P}(\mathrm{succ})\geq\gamma|\mathbb{P}(\mathrm{succ})\geq\gamma')\\
        &+\mathbb{E}[\mathbb{P}(\mathrm{succ})|\gamma'\leq\mathbb{P}(\mathrm{succ})< \gamma]\mathbb{P}(\gamma'\leq\mathbb{P}(\mathrm{succ})< \gamma|\mathbb{P}(\mathrm{succ})\geq\gamma')
    \end{align}

    and $\mathbb{E}[\mathbb{P}(\mathrm{succ})|\mathbb{P}(\mathrm{succ})\geq\gamma)]>\mathbb{E}[\mathbb{P}(\mathrm{succ})|\gamma'\leq\mathbb{P}(\mathrm{succ})<\gamma]$. Therefore,

    \begin{equation}
        \mathbb{E}[\mathbb{P}(\mathrm{succ})|\mathbb{P}(\mathrm{succ})\geq\gamma']\leq\mathbb{E}[\mathbb{P}(\mathrm{succ})|\mathbb{P}(\mathrm{succ})\geq\gamma],
    \end{equation}

    that is, $\mathbb{P}_{\mathrm{succ}}^{\gamma'}\leq \mathbb{P}_{\mathrm{succ}}^{\gamma}$ with equality if and only if $\mathbb{P}(\gamma'\leq\mathbb{P}(\mathrm{succ})< \gamma|\mathbb{P}(\mathrm{succ})\geq\gamma')=0$.

\end{proof}

Now, we use this lemma to prove Theorem 1.

\begin{proof}
    First we note that by virtue of MLD, we can write

    \begin{equation}\label{eq:psucc}
        \mathbb{P}_{\mathrm{succ}} = [\max_{\sigma\in\mathcal{P}^{\otimes k}}Z_{\sigma E(S)}]_{S_{\mathrm{accept}}}.
    \end{equation}
    
    As above, we define $\{S_{\mathrm{accept}}=S|\max_{\sigma\in\mathcal{P}^{\otimes k}}Z_{\sigma E(S)}>\gamma\}$. Then, for any other choice of post-selection which partitions $S_{\mathrm{accept}}$ with equivalent cardinality but different elements, at least one element of $S_{\mathrm{accept}}$ must have $\max_{\sigma\in\mathcal{P}^{\otimes k}Z_{\sigma E(S)}}<\gamma$ strictly. Therefore, by Eq. \ref{eq:psucc}, $\mathbb{P}_{\mathrm{succ}}$ must be strictly smaller for any other post-selection. Hence, this post-selection is optimal.
\end{proof}

\section{S3. Analytic and numerical investigation of fully post-selected logical thresholds}
The Hamiltonian of the toric code under a bit-flip error channel is

\begin{equation}\label{eq:rbim}
    H_{E} = -\sum_{i,j}(-2K + J\llbracket Z,E_{i,j,\leftrightarrow}\rrbracket s_{i,j}^{X}s_{i-1,j}^{X}+ J\llbracket Z,E_{i,j,\updownarrow}\rrbracket s_{i,j}^{X}s_{i,j-1}^{X}),
\end{equation}

in which the Nishimori conditions relate the temperature and bit-flip probability through

\begin{equation}\label{eq:rbim_nishimori}
\begin{split}
    K &= \frac{1}{2\beta}\log(\frac{1}{p(1-p)}),\\
    J &=\frac{1}{2\beta}\log(\frac{1-p}{p}).
\end{split}
\end{equation}

The Hamiltonian of the toric code under a depolarizing error channel is

\begin{equation}\label{eq:eightvert}
\begin{split}
    H_{E} = -\sum_{i,j}(-K &+ J\llbracket X,E_{i,j,\leftrightarrow}\rrbracket s_{i,j}^{Z}s_{i,j+1}^{Z}+ J\llbracket Z,E_{i,j,\updownarrow}\rrbracket s_{i,j}^{Z}s_{i+1,j}^{Z}+ J\llbracket Z,E_{i,j,\updownarrow}\rrbracket s_{i,j}^{Z}s_{i-1,j}^{Z}\\
    &+ J\llbracket X,E_{i,j,\updownarrow}\rrbracket s_{i,j}^{Z}s_{i,j-1}^{Z}+ J\llbracket X,E_{i,j,\updownarrow}\rrbracket s_{i,j}^{Z}s_{i,j+1}^{Z}s_{i,j}^{Z}s_{i-1,j}^{Z}+ J\llbracket X,E_{i,j,\updownarrow}\rrbracket s_{i,j}^{Z}s_{i+1,j}^{Z}s_{i,j}^{Z}s_{i,j-1}^{Z}),
\end{split}
\end{equation}

where the Nishimori conditions are
\begin{equation}\label{eq:eightvert_nishimori}
\begin{split}
    \beta K = \frac{1}{2}\text{log}\frac{27}{p^{3}(1-p)},\\
    \beta J = \frac{1}{4}\text{log}\frac{3(1-p)}{p}.
\end{split}
\end{equation}

We numerically verify the upper bounds on the conditional logical thresholds for depolarizing and bit-flip channels computed in the main article. To perform MLD, we use a matrix product state (MPS)-based tensor network (TN) decoder \cite{BravyiSergey2014Eafm} implemented in qecsim \cite{qecsim}. We use full contraction of a TN decoder, computing the coset probabilities of a planar surface code of distance $d$ when the syndrome is trivial. Thresholds of planar surface codes and toric codes are equivalent, as the thermodynamics are unchanged by the boundaries when $n\rightarrow\infty$. In Figs. \ref{fig:bitflip-th} and \ref{fig:depolarizing-th} we plot the posterior confidence that the identity coset will correct successfully (i.e., that by doing nothing, no error has occurred on the logical qubit) for bit-flip and phase-flip channels, respectively. In this fully post-selected regime, i.e., when $\gamma=1$, the pseudothreshold (or breakeven) of the code is equal to the threshold.

\begin{figure}[h]
    \centering
    \includegraphics[width=0.5\columnwidth]{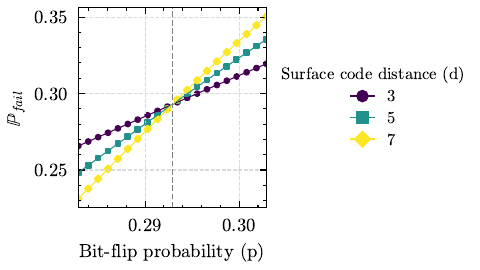}
    \caption{Probability of decoding failure by doing nothing when syndrome is trivial for a bit-flip channel. The dashed line indicates the theoretical threshold value computed in the main article.}
    \label{fig:bitflip-th}
\end{figure}
\begin{figure}[h]
    \centering
    \includegraphics[width=0.5\columnwidth]{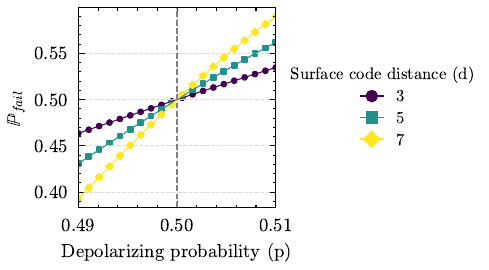}
    \caption{Probability of decoding failure by doing nothing when syndrome is trivial for a depolarizing channel. The dashed line indicates the theoretical threshold value computed in the main article.}
    \label{fig:depolarizing-th}
\end{figure}

\section{Abort thresholds under optimal post-selection}
Below logical threshold of a code, 

\begin{equation}
    \lim_{n\rightarrow\infty}\mathbb{P}(\max_{E}\mathbb{P}(\bar{E})\geq 1-\delta)=1\quad \forall \delta\neq 0,
\end{equation}

that is, the maximum coset probability approaches 1 almost surely. Similarly, above logical threshold,

\begin{equation}
    \lim_{n\rightarrow\infty}\mathbb{P}(\max_{E}\mathbb{P}(\bar{E})\leq 1/K - \eta)=1\quad \forall \eta\neq 0,
\end{equation}

that is, the maximum coset probability approaches 1/K (for a logical algebra of dimension K) almost surely. Then in the thermodynamic limit, for any $\gamma\in(1/K,1)$, for $p<p_{\mathrm{th}}$, we will almost surely never abort, and for $p>p_{\mathrm{th}}$, we will almost surely always abort. Therefore, under optimal post-selection, the abort threshold is independent of the post-selection parameter, and fixed at the code's non-post-selected logical threshold.

\section{S4. Details of numerical simulations}

If sampling at sufficiently high $p$ relative to the abort threshold, the probability of aborting tends towards $1$, forbidding simple rejection sampling methods. We use the Metropolis-Hastings algorithm for a Markov-Chain Monte Carlo (MCMC) method \cite{HastingsW.K.1970MCsm} to sample the conditional probability distribution $\mathbb{P}(E|m(S)>m)$, that is probabilities of errors sampled under an i.i.d. bit-flip channel of probability $p$, conditioned on the syndrome having magnetization above the post-selection criterion $m$. The Metropolis-Hastings algorithm can be used, as the i.i.d. error channels allow the detailed balance condition to be trivially satisfied, and local steps (i.e., flipping single spins) used by such a sampler retains ergodicity, as different logical sectors can be coupled without aborting. Decoding is performed using a minimum-weight perfect matching algorithm through PyMatching \cite{higgott2023sparse}.

\section{S5. Four phases in the thermodynamic limit}

Diagrammatically, we can plot in Fig. \ref{fig:regimes} what the four thermodynamic regimes correspond to in the thermodynamic limit with respect to conditional logical failure probabilities and abort probabilities. These probabilities should be considered such that the probability of failure almost surely approaches 0 or $\frac{K-1}{K}$ for a logical algebra of dimension $K$, and the probability of aborting approaches $0$ or $1$ almost surely. For example, point 1 corresponds to being below both thresholds, while point 3 corresponds to being above both thresholds. Points 2 and 3 correspond to being above one threshold, while below the other. Dashed lines indicate possible transitions as one increases post-selection. For example, suppose we begin above logical threshold without any post-selection. For $\gamma<0.5$, we are below abort threshold for all $p\in[0,1]$. In this case, we begin at point 4. As we increase post-selection, we move horizontally to the right in the phase diagram, until we possibly become below both thresholds, where we follow the dashed line to point 1. If we increase post-selection further, we may end up above the abort threshold, but below the logical threshold, arriving at point 2. Conversely, we may also begin at a sufficiently high physical error rate, where instead we make the transition from point 4 to point 3 and remain above both thresholds.

\begin{figure}
    \centering
\includegraphics[width=0.5\linewidth]{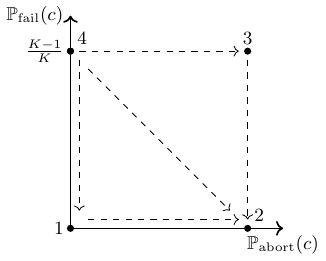}
    \caption{Probability of conditional failure and probability of aborting in the thermodynamic limit. Four regimes exist which depend on the physical error rates and the post-selection utilized.}
    \label{fig:regimes}
\end{figure}

\section{S6. Useful scalable post-selection}

Our conjecture on the usefulness of the shaded region in the phase diagram (i.e., located within Phase I) is based on the finite-size scaling behaviors observed in toric codes without post-selection. Specifically, the logical failure rate in standard toric codes follows the scaling ansatz:
\begin{equation}
    \mathbb{P}_{\mathrm{fail}}=f((p-p_{c})d^{\frac{1}{\nu}}),
\end{equation}
where $p_{c}$ is the logical threshold, $d$ is the code distance, and $\nu$ is the critical exponent governing the scaling near the threshold \cite{WatsonFernHE2014Lers}. Given a monotonic scaling function $f$, this scaling ansatz implies that increasing the code distance $d$ systematically reduces $\mathbb{P}_{\mathrm{fail}}$, provided $p<p_{c}$.

In the post-selected regime, we observe a conditional threshold $p_{c}^{\mathrm{cond}}$ that can be increased beyond the unconditioned threshold $p_{c}$ by tuning the post-selection parameter $\gamma$. This is evident in the boundary between Phases II and III where $p_{c}^{\mathrm{cond}}>p_{c}$. If the same scaling behavior holds for post-selected quantum error correction, the logical failure rate $\mathbb{P}_{\mathrm{fail}}$ can be decreased by either:
\begin{enumerate}
    \item Increasing the code distance $d$, which reduces logical errors via standard error correction mechanisms, or
    \item Adjusting the post-selection parameter $\gamma$ which increases $p_{c}^{\mathrm{cond}}$, and moves the system further away from the critical point for a given physical error rate $p$.
\end{enumerate}

In finite-size codes, these effects combine to provide a pathway for reducing logical failure rates in practical implementations. The shaded region represents the phase in which method 2 above is achieved, that is, $p_{c}^{\mathrm{cond}}>p_{c}$, decreasing the logical failure rate, while being below both logical and abort thresholds. All of Phase I is characterized by being below both abort and logical thresholds. However, the non-shaded region of Phase I has a conditional logical threshold equal to the unconditioned threshold, $p_{c}^{\mathrm{cond}}=p_{c}$. We claim the shaded region is the useful region for post-selection because of the increased threshold while remaining in the scalable Phase I. Future work should explore how the critical exponent $\nu$ varies with the post-selection parameter $\gamma$. For $\gamma=1$, where the fully post-selected threshold corresponds to the phase transition of the non-disordered Ising model, the critical exponent is expected to take the Ising model's value of $\nu=1$. An open question remains whether this variation in $\nu$ with $\gamma$ is continuous or discontinuous, which merits further investigation.

\section{S7. Effectiveness of heuristic post-selection rule}

The optimal post-selection criterion operates independently of the underlying noise processes. Specifically, as long as the error channels are Markovian--enabling the calculation of logical error coset probabilities--Theorem 1 remains valid. The heuristic post-selection rule, though less rigorous, similarly functions independently of the noise model under one key condition: the density of violated stabilizers must be anti-correlated with the logical success probability, conditioned on the observed syndrome. This condition is inherently satisfied by all i.i.d. error channels acting on surface codes.

More broadly, in scenarios involving independent but not necessarily identically distributed error channels, the MWPM decoder represents violated stabilizers as vertices in a graph, with edges weighted by error probabilities inferred from the noise model. As the number of violated stabilizers increases, the number of possible perfect matchings grows combinatorially. Consequently, the likelihood of a minimum-weight matching inducing a logical operator also increases, as the larger set of error configurations creates more opportunities for errors to align in a manner that spans a logical operator. This, in turn, diminishes the decoder’s ability to reliably identify the correct logical coset, reducing the conditional logical success probability.

These insights highlight the critical role of the anti-correlation between violated stabilizer density and logical success probability in ensuring the effectiveness of the heuristic post-selection rule. While the rule provides a flexible and effective approach to improving logical success probabilities across various error models, its applicability may be limited in the presence of highly correlated or adversarial noise channels. Addressing such challenges may require alternative strategies or refinements to the heuristic approach. Future work could focus on extending this framework to handle non-Markovian noise models or correlated errors, potentially expanding its scope and impact in practical quantum error correction implementations.

\end{document}